\documentclass{article}

\usepackage{url}
\usepackage{array}
\usepackage{multirow}
\usepackage{amsmath}
\usepackage{amssymb}
\usepackage{amsthm}
\usepackage[all]{xy}

\newcommand{\maps}{\colon}
\newcommand{\R}{\mathbb{R}}
\newcommand{\C}{\mathbb{C}}
\newcommand{\ip}[2]{\langle #1, #2 \rangle}

\DeclareMathOperator{\HilbI}{\mathbf{Hilb_I}}
\DeclareMathOperator{\HilbE}{\mathbf{Hilb_E}}
\DeclareMathOperator{\SympI}{\mathbf{Symp_I}}
\DeclareMathOperator{\SympE}{\mathbf{Symp_E}}

\DeclareMathOperator{\Hom}{Hom}

\theoremstyle{plain}
\newtheorem{prop}{Proposition}
\theoremstyle{definition}
\newtheorem{defn}{Definition}
\newtheorem{exmp}{Example}

\title{Limitations on Cloning in Classical Mechanics}
\date{}
\author{Aaron Fenyes\\
Department of Mathematics\\
University of Texas at Austin\\
{\tt afenyes@math.utexas.edu}}

\begin{document}
\maketitle
\begin{abstract}
In this paper, we show that a result precisely analogous to the traditional quantum no-cloning theorem holds in classical mechanics. This classical no-cloning theorem does not prohibit classical cloning, we argue, because it is based on a too-restrictive definition of cloning. Using a less popular, more inclusive definition of cloning, we give examples of classical cloning processes. We also prove that a cloning machine must be at least as complicated as the object it is supposed to clone.
\end{abstract}
\maketitle
\section{Introduction}
In the three decades since its discovery by Wootters and Zurek, the no-cloning principle has come to be regarded as one of the most important basic results in quantum mechanics. Its mathematical formulations, its consequences, and its limitations have been studied extensively \cite{qc-review}.

Interest in cloning has also spread to fields outside quantum mechanics. Daffertshofer, Plastino, and Plastino, for example, proved a no-cloning theorem in classical statistical mechanics \cite{daff-plas-plas}, although its interpretation has been the subject of some dispute \cite{walker}. Abramsky, taking a more abstract approach, proved a no-cloning theorem that applies to any physical theory whose ``native category'' is compact closed \cite{abramsky}.

Category theory is a good setting for the study of cloning because it provides a very general way to talk about the properties of composite systems, which control, to a large extent, the possibility or impossibility of cloning. In the categories typically used as settings for quantum mechanics, statistical mechanics, and classical mechanics, composition of physical systems is represented by a ``monoid operation'' that obeys certain rules. In the category of Hilbert spaces and linear isometries---a natural setting for quantum mechanics---the monoid operation is just the usual tensor product of Hilbert spaces. The quantum no-cloning principle is intuitively related to the fact that this monoid operation is {\em non-cartesian}---that is, it does not act like the familiar cartesian product of sets \cite{baez-notes}.

In the category of Poisson manifolds and Poisson maps---a natural setting for classical mechanics---the monoid operation combines the Poisson structures of two manifolds into a new Poisson structure on the cartesian product of the manifolds. Baez recently pointed out that although this monoid operation is cartesian when one forgets about the manifolds' Poisson structures, it is non-cartesian when the Poisson structures are taken into account. This observation led Baez to suggest that some kind of no-cloning principle might exist in classical mechanics \cite{baez-notes}. In the category of symplectic manifolds \cite{symplectic} and symplectic maps (a slightly less general setting for classical mechanics), this counterintuitive conjecture turns out to be correct: in Section~\ref{problem}, we present a true statement about symplectic mechanics that is clearly analogous to a popular formulation of the quantum no-cloning principle.

How does one reconcile this classical no-cloning result with the deep-seated intuition that it is always possible to clone an unknown classical state? We argue in Section~\ref{cause} that the classical no-cloning result from Section~\ref{problem} does not actually prohibit classical cloning, because it is based on a definition of cloning that does not allow a cloning process to involve any physical system except the object to be cloned and the raw material for the clone. Under the more inclusive definition of cloning used in \cite{daff-plas-plas} and \cite[Section~I\thinspace A]{qc-review}, quantum cloning remains impossible, but classical cloning becomes possible for many systems, as discussed in Sections \ref{eq.nc.subs} and \ref{cloning-examples}.

The definition of cloning used in \cite{daff-plas-plas} and \cite[Section~I\thinspace A]{qc-review} allows a cloning process to involve a physical system other than the object to be cloned and the raw material for the clone---that is, a cloning machine. In Section~\ref{size.sec}, we show that the dimension of the phase space of a cloning machine must be at least as great as the dimension of the phase space of the object to be cloned.

With the course of the paper now laid out, let us go back and see in more detail how Baez's classical no-cloning conjecture is confirmed.
\section{A problematic result, and its cause}
\subsection{The problem: classical no-cloning}\label{problem}
In quantum mechanics, the following result is often said to prohibit the cloning of unknown states:
\begin{prop}[Traditional no-cloning theorem---quantum]\label{tq.nc}
Let $H$ be a complex Hilbert space with dimension greater than zero. There cannot exist a unit vector $\beta \in H$ and a unitary map\footnote{This proposition remains true even if $U$ is only required to be a linear isometry (see Definition~\ref{lin.iso}).} $U \maps H \otimes H \to H \otimes H$ such that $U(\psi \otimes \beta) = \psi \otimes \psi$ for all unit vectors $\psi \in H$.
\end{prop}
This result, however, cannot be advertised as a major difference between quantum and classical mechanics, because a precisely analogous result\footnote{The analogy will be made precise in Section~\ref{cat-cloning}.} holds in classical mechanics.
\begin{prop}[Traditional no-cloning theorem---classical]\label{tc.nc}
Let $M$ be a symplectic manifold with dimension greater than zero. There cannot exist a point $b \in M$ and a symplectomorphism\footnote{This proposition remains true even if $\phi$ is only required to be a symplectic map (see Definition~\ref{symp.map}).} $\phi \maps M \times M \to M \times M$ such that $\phi(x, b) = (x, x)$ for all $x \in M$.
\end{prop}
\begin{proof}
This result is a special case of Proposition~\ref{size} (in Section~\ref{size.sec}). It can be obtained by setting the dimension of $N$ in Proposition~\ref{size} to zero.
\end{proof}
Intuition tells us (through its spokesperson, Simon Saunders) that it is easy to clone an unknown classical state---all we have to do is measure the state, and then make a copy of it. Proposition~\ref{tc.nc}, on the other hand, seems to say that we {\em cannot} clone an unknown classical state. What is responsible for this discrepancy?
\subsection{The cause: our definition of cloning}\label{cause}
Implicit in our interpretation of Propositions \ref{tq.nc} and \ref{tc.nc} is a definition\footnote{There appear to be {\em two} definitions here, but they are both special cases of a single, more general definition, as explained in Section~\ref{cat-cloning}.} of cloning:
\begin{defn}[Traditional definition of cloning---quantum]\label{tq.def}
Let $H$ be a complex Hilbert space. A {\em cloning process} for $H$ consists of a unit vector $\beta \in H$ and a unitary map\footnote{The propositions that depend on this definition remain true even if $U$ is only required to be a linear isometry (see Definition~\ref{lin.iso}).} $U \maps H \otimes H \to H \otimes H$ such that $U(\psi \otimes \beta) = \psi \otimes \psi$ for all unit vectors $\psi \in H$.
\end{defn}
\begin{defn}[Traditional definition of cloning---classical]\label{tc.def}
Let $M$ be a symplectic manifold. A {\em cloning process} for $M$ consists of a point $b \in M$ and a symplectomorphism\footnote{The propositions that depend on this definition remain true even if $\phi$ is only required to be a symplectic map (see Definition~\ref{symp.map}).} $\phi \maps M \times M \to M \times M$ such that $\phi(x, b) = (x, x)$ for all $x \in M$.
\end{defn}
The confusion at the end of the previous section arose because this definition is too restrictive, excluding processes that most people would intuitively classify as cloning.

To understand why, let us see how a cloning process of the kind described in Definition~\ref{tc.def} would act on a familiar system. Suppose $M$ is the phase space of a set of wooden blocks, and the point $b \in M$ represents a situation in which the blocks are laid out nicely on the floor. We have built a tower out of the blocks; this state of affairs is represented by the point $x \in M$. To get a copy of the tower, we bring in a second set of blocks and lay it out in state $b$ next to the tower. The phase space of the two sets of blocks is $M \times M$, and our current situation is represented by the point $(x, b) \in M \times M$. We now wait for ten seconds as the blocks interact with each other; the evolution of the system over this period is described by the Poisson map $\phi$. At the end of the ten seconds, the system is in state $(x, x)$; the blocks have arranged themselves into a pair of towers identical to the tower we started with.

In this story, the cloning process involved only the object to be cloned and the raw material for the clone, interacting according to the laws of nature. Few people, however, would try to clone a tower of blocks that way. A more practical approach would be to introduce a third system: a machine equipped with sensors and manipulators. The machine would start out in a special ``ready'' state; allowed to evolve according to the laws of nature, it would use its sensors to determine the configuration of the tower, and then use its manipulators to build a copy out of the second set of blocks. Crucially, the machine would not be expected to end up back in the ready state; in fact, under most circumstances, it would be physically impossible for the machine to end up back in the ready state. Suppose, for example, the machine were powered by a clock spring. If the center-of-mass of the tower were higher than the center-of-mass of the blocks laid out in state $b$, the spring would have to wind down in order to lift the blocks.
\section{An expanded definition, and its consequences}
\subsection{An expanded definition of cloning}
Instead of the traditional definition of cloning, let us work with the definition used in \cite{daff-plas-plas} and \cite[Section~I\thinspace A]{qc-review}, which encompasses processes like the one described in the previous paragraph.
\begin{defn}[Expanded definition of cloning---quantum]\label{eq.def}
Let $H$ be a complex Hilbert space. A {\em cloning process} for $H$ consists of a unit vector $\beta \in H$, a complex Hilbert space $K$, a unit vector $\rho \in K$, and a unitary map\footnote{The propositions that depend on this definition remain true even if $U$ is only required to be a linear isometry (see Definition~\ref{lin.iso}).} $U \maps H \otimes H \otimes K \to H \otimes H \otimes K$ such that $U(\psi \otimes \beta \otimes \rho) = \psi \otimes \psi \otimes f(\psi)$ for all unit vectors $\psi \in H$, where $f$ is some function from the unit sphere of $H$ to the unit sphere of $K$.
\end{defn}
\begin{defn}[Expanded definition of cloning---classical]\label{ec.def}
Let $M$ be a symplectic manifold. A {\em cloning process} for $M$ consists of a point $b \in M$, a symplectic manifold $N$, a point $r \in N$, and a symplectomorphism\footnote{The propositions that depend on this definition remain true even if $\phi$ is only required to be a symplectic map (see Definition~\ref{symp.map}).} $\phi \maps M \times M \times N \to M \times M \times N$ such that $\phi(x, b, r) = (x, x, f(x))$ for all $x \in M$, where $f \maps M \to N$ is some function.
\end{defn}
\subsection{Quantum no-cloning}\label{eq.nc.subs}
Under the expanded definition of cloning from \cite{daff-plas-plas} and \cite[Section~I\thinspace A]{qc-review}, quantum cloning is still impossible for all but the most trivial systems.
\begin{prop}[{\cite[Section~I\thinspace A]{qc-review}}]
Let $H$ be a complex Hilbert space with dimension greater than one. There cannot exist a cloning process for $H$ of the kind described in Definition~\ref{eq.def}.
\end{prop}
\begin{proof}
Suppose there does exist such a cloning process. Since $\dim H > 1$, we can find two unit vectors $\psi, \tilde{\psi} \in H$ such that $0 < |\ip{\psi}{\tilde{\psi}}| < 1$. Since $U$ is unitary,
\begin{align}
\ip{\psi \otimes \beta \otimes \rho}{\tilde{\psi} \otimes \beta \otimes \rho} & = \ip{\psi \otimes \psi \otimes f(\psi)}{\tilde{\psi} \otimes \tilde{\psi} \otimes f(\tilde{\psi})} \notag\\
\ip{\psi}{\tilde{\psi}}\ip{\beta}{\beta}\ip{\rho}{\rho} & = \ip{\psi}{\tilde{\psi}}\ip{\psi}{\tilde{\psi}}\ip{f(\psi)}{f(\tilde{\psi})} \notag\\
1 & = \ip{\psi}{\tilde{\psi}}\ip{f(\psi)}{f(\tilde{\psi})}.\label{eq.contra}
\end{align}
Since $|\ip{\psi}{\tilde{\psi}}| < 1$, Equation~\ref{eq.contra} implies that $|\ip{f(\psi)}{f(\tilde{\psi})}| > 1$. Since $f(\psi)$ and $f(\tilde{\psi})$ are unit vectors, this contradicts the Cauchy-Schwarz inequality.
\end{proof}
\subsection{Some classical cloning processes}\label{cloning-examples}
The expanded definition of cloning from \cite{daff-plas-plas} and \cite[Section~I\thinspace A]{qc-review}, unlike the traditional definition, allows the existence of cloning processes for many (and perhaps all) classical systems. The easiest way to prove this is to give some examples.
\begin{exmp}\label{basic}
Let $M$ and $N$ be $\R^2$ equipped with the standard symplectic form
\[ \omega(x, \tilde{x}) = x \cdot \left[\begin{array}{rr} 0 & 1 \\ -1 & 0 \end{array}\right] \tilde{x}, \]
where $\cdot$ is the usual inner product on $\R^2$. Let $b \in M$ and $r \in N$ be zero, and let $\phi \maps M \times M \times N \to M \times M \times N$ be the map
\[ \phi(x_1, x_2, x_3) = \left[\begin{array}{rr|rr|rr}
1 & 0 & 1 & 0 & 0 & 0 \\
0 & 1 & 0 & 0 & 0 & -1 \\
\hline
1 & 0 & -1 & 0 & 1 & 0 \\
0 & 1 & 0 & -1 & 0 & -1 \\
\hline
1 & 0 & 0 & 0 & 1 & 0 \\
0 & -1 & 0 & 1 & 0 & 2
\end{array}\right]
\left[\begin{array}{r}
\multirow{2}{*}{$x_1$} \\ \vphantom{0} \\
\hline
\multirow{2}{*}{$x_2$} \\ \vphantom{0} \\
\hline
\multirow{2}{*}{$x_3$} \\ \vphantom{0}
\end{array}\right]. \]
This gives a cloning process on $M$ of the kind described in Definition~\ref{ec.def}.
\end{exmp}
\begin{proof}
It is clear that $\phi(x, b, r) = (x, x, Fx)$, where
\[ F = \left[\begin{array}{rr} 1 & 0 \\ 0 & -1 \end{array}\right]. \]
To prove we have a cloning process, all we have to do is show that $\phi$ is a Poisson map. Since $M \times M \times N$ is a symplectic vector space, and $\phi$ is linear, it is enough to show that $\phi$ preserves the symplectic structure of $M \times M \times N$ at $(0, 0, 0)$.

The symplectic form on $M \times M \times N$ is
\begin{align}
\xi[(x_1, x_2, x_3), (\tilde{x}_1, \tilde{x}_2, \tilde{x}_3)] & = \omega(x_1, \tilde{x}_1) + \omega(x_2, \tilde{x}_2) + \omega(x_3, \tilde{x}_3) \notag\\
& = \left[\begin{array}{r}
\multirow{2}{*}{$x_1$} \\ \vphantom{0} \\
\hline
\multirow{2}{*}{$x_2$} \\ \vphantom{0} \\
\hline
\multirow{2}{*}{$x_3$} \\ \vphantom{0}
\end{array}\right] \cdot
\left[\begin{array}{rr|rr|rr}
0 & 1 & 0 & 0 & 0 & 0 \\
-1 & 0 & 0 & 0 & 0 & 0 \\
\hline
0 & 0 & 0 & 1 & 0 & 0 \\
0 & 0 & -1 & 0 & 0 & 0 \\
\hline
0 & 0 & 0 & 0 & 0 & 1 \\
0 & 0 & 0 & 0 & -1 & 0
\end{array}\right]
\left[\begin{array}{r}
\multirow{2}{*}{$\tilde{x}_1$} \\ \vphantom{0} \\
\hline
\multirow{2}{*}{$\tilde{x}_2$} \\ \vphantom{0} \\
\hline
\multirow{2}{*}{$\tilde{x}_3$} \\ \vphantom{0}
\end{array}\right].\label{giant}
\end{align}
For compactness, we will refer to the giant matrix in Equation~\ref{giant} as $\Xi$.

The linear map $\phi$ is a symplectomorphism if and only if the (equivalent) conditions
\begin{align*}
\xi[d\phi(x_1, x_2, x_3), d\phi(\tilde{x}_1, \tilde{x}_2, \tilde{x}_3)] & = \xi[(x_1, x_2, x_3), (\tilde{x}_1, \tilde{x}_2, \tilde{x}_3)] \\
\phi(x_1, x_2, x_3) \cdot \Xi \phi(\tilde{x}_1, \tilde{x}_2, \tilde{x}_3) & = (x_1, x_2, x_3) \cdot \Xi (\tilde{x}_1, \tilde{x}_2, \tilde{x}_3) \\
(x_1, x_2, x_3) \cdot \phi^\top \Xi \phi(\tilde{x}_1, \tilde{x}_2, \tilde{x}_3) & = (x_1, x_2, x_3) \cdot \Xi (\tilde{x}_1, \tilde{x}_2, \tilde{x}_3)
\end{align*}
hold for all $(x_1, x_2, x_3)$ and $(\tilde{x}_1, \tilde{x}_2, \tilde{x}_3)$ in $M \times M \times N$. (Here, $\phi^\top$ denotes the matrix transpose of $\phi$; as a linear map, $\phi^\top$ is the adjoint of $\phi$ with respect to the inner product $\cdot$.) It is easy to verify that $\phi^\top \Xi \phi = \Xi$, so the last condition does always hold. Therefore, $\phi$ is a symplectomorphism, giving us a cloning process on $M$.
\end{proof}
Now that we have one example of a classical cloning process, we can use the following proposition to construct more.
\begin{prop}\label{products}
If we have a cloning process for the symplectic manifold $M_1$ and a cloning process for the symplectic manifold $M_2$, we can explicitly construct a cloning process for $M_1 \times M_2$.
\end{prop}
\begin{proof}
Suppose that for each $i \in \{1, 2\}$ we have a point $b_i \in M_i$, a symplectic manifold $N_i$, a point $r_i \in N_i$, and a symplectomorphism $\phi_i \maps M_i \times M_i \times N_i \to M_i \times M_i \times N_i$ such that $\phi_i(x, b_i, r_i) = (x, x, f_i(x))$ for all $x \in M_i$, where $f_i \maps M_i \to N_i$ is some function. Let $\phi$ be the symplectomorphism from $(M_1 \times M_1 \times N_1) \times (M_2 \times M_2 \times N_2)$ to itself given by
\[ \phi[(x_1, \tilde{x}_1, y_1), (x_2, \tilde{x}_2, y_2)] = [\phi_1(x_1, \tilde{x}_1, y_1), \phi_2(x_2, \tilde{x}_2, y_2)], \]
and observe that $\phi[(x_1, b_1, r_1), (x_2, b_2, r_2)] = [(x_1, x_1, f_1(x_1)), (x_2, x_2, f_2(x))]$ for all $(x_1, x_2) \in M_1 \times M_2$. The canonical isomorphism between $(M_1 \times M_1 \times N_1) \times (M_2 \times M_2 \times N_2)$ and $(M_1 \times M_2) \times (M_1 \times M_2) \times (N_1 \times N_2)$ then gives us a cloning process for $M_1 \times M_2$.
\end{proof}
With the help of Proposition~\ref{products}, we can extend Example~\ref{basic} to symplectic vector spaces of any finite dimension.
\begin{exmp}\label{general}
Let $M$ be a finite-dimensional symplectic vector space. By Darboux's theorem, $M$ is isomorphic to $\bigoplus^n \R^2$ for some $n$. The $n = 0$ case is trivial. When $n > 0$, Example~\ref{basic} gives us a cloning process for each copy of $\R^2$, and Proposition~\ref{products} shows us how to combine these cloning processes into a cloning process for $\bigoplus^n \R^2$.
\end{exmp}

The dynamics of a classical system can be turned into the dynamics of a statistical system by specifying the initial state of the system probabilistically instead of exactly. One might therefore try to use the classical cloning processes described here to construct a statistical cloning process that would violate the no-cloning theorem proven in \cite{daff-plas-plas}. To see where this scheme may fail (and, in light of \cite{daff-plas-plas}, must fail), observe that the cloning processes described here only work when the raw material for the clone begins in the blank state, $b$, and the cloning machine begins in the ready state, $r$. The hypotheses of \cite{daff-plas-plas} exclude singular probability distributions like the Dirac delta distribution, so there is no way to ensure that the raw material and the cloning machine start off in the right state.
\subsection{Minimum size of a classical cloning machine}\label{size.sec}
In the product manifold $M \times M \times N$ that appears in Definition~\ref{ec.def}, the first factor of $M$ should be interpreted as the phase space of the object to be cloned, the second factor of $M$ should be interpreted as the phase space of the raw material for the clone, and the the factor of $N$ should be interpreted as the phase space of any extra machinery to be used in the cloning process. With a little thought, we can deduce something interesting about the nature of this extra machinery.
\begin{prop}\label{size}
For any cloning process of the kind described in Definition~\ref{ec.def}, the dimension of $N$ must be greater than or equal to the dimension of $M$.
\end{prop}
This result has a clear physical interpretation: it says that in symplectic mechanics, a cloning machine must be at least as complicated as the object to be cloned.
\begin{proof}[Proof of Proposition~\ref{size}]
Pick any $x \in M$. Because
\[ \phi(\tilde{x}, b, r) = (\tilde{x}, \tilde{x}, f(\tilde{x})) \]
for all $\tilde{x} \in M$,
\[ d\phi_{(x, b, r)}(v, 0, 0) = (v, v, df_x(v)) \]
for all $v \in T_x M$. (Here, $d\phi_{(x, b, r)}$ denotes the total derivative of $\phi$ at the point $(x, b, r)$, and $df_x$ denotes the total derivative of $f$ at the point $x$.) For convenience, let us write $df_x$ as $F$. Keep in mind that $F$ is a linear map from $T_xM$ to $T_{f(x)}N$.

Let $\omega$ and $\sigma$ be the symplectic forms on $M$ and $N$ respectively, and let $\xi$ be the standard symplectic form on $M \times M \times N$, defined in terms of $\omega$ and $\sigma$ as
\[ \xi_{(x_1, x_2, x_3)}[(v_1, v_2, v_3), (\tilde{v}_1, \tilde{v}_2, \tilde{v}_3)] = \omega_{x_1}(v_1, \tilde{v}_1) + \omega_{x_2}(v_2, \tilde{v}_2) + \sigma_{x_3}(v_3, \tilde{v}_3). \]
Because $\phi$ is a symplectomorphism,
\begin{align}
\xi_{(x, b, r)}[(v, 0 ,0), (\tilde{v}, 0, 0)] & = \xi_{\phi(x, b, r)}[d\phi_{(x, b, r)}(v, 0, 0), d\phi_{(x, b, r)}(\tilde{v}, 0, 0)] \notag\\
\xi_{(x, b, r)}[(v, 0 ,0), (\tilde{v}, 0, 0)] & = \xi_{(x, x, f(x))}[(v, v, Fv), (\tilde{v}, \tilde{v}, F\tilde{v})] \notag\\
\omega_x(v, \tilde{v}) & = \omega_x(v, \tilde{v}) + \omega_x(v, \tilde{v}) + \sigma_{f(x)}(Fv, F\tilde{v}) \notag\\
-\omega_x(v, \tilde{v}) & = \sigma_{f(x)}(Fv, F\tilde{v}).\label{size.contra}
\end{align}

Recall that $F$ is a linear map from $T_xM$ to $T_{f(x)}N$. If the dimension of $N$ is less than the dimension of $M$, the kernel of $F$ is non-trivial---that is, there exists some non-zero vector $w \in T_xM$ for which $Fw = 0$. Equation~\ref{size.contra} then implies that $\omega_x(w, \tilde{v}) = 0$ for all $\tilde{v} \in T_xM$, contradicting the fact that $\omega_x$ is non-degenerate. To avoid this contradiction, the dimension of $N$ must be greater than or equal to the dimension of $M$.
\end{proof}
\section{A categorical approach to cloning}
\subsection{Cloning in symmetric monoidal categories}\label{cat-cloning}
In Section~\ref{cause}, we presented Definition~\ref{tc.def} as a classical analogue of the traditional definition of quantum cloning (Definition~\ref{tq.def}). We justified the analogy by claiming that both definitions are special cases of a single, more general definition. Before we state that more general definition, let us conduct a brief review of the category theory it is based on.

In the quantum theory of closed systems, physical systems are represented by complex Hilbert spaces, and physical transformations are represented by unitary maps.\footnote{When open systems are taken into consideration, some transformations have to be represented by non-unitary maps, leading to a more complicated categorical setting. One possible choice is the category whose objects and arrows are complex Hilbert spaces and bounded linear maps, discussed in \cite{quandaries}.} A natural mathematical setting for quantum mechanics is therefore the category whose objects and arrows are complex Hilbert spaces and unitary maps. We will refer to this category as $\HilbI$.

In classical mechanics, physical systems are represented by Poisson manifolds, and physical transformations are represented by Poisson maps. A natural mathematical setting for classical mechanics is therefore the category whose objects and arrows are Poisson manifolds and Poisson maps. In this paper, we have been working in a slightly less general setting: the category whose objects and arrows are symplectic manifolds and symplectomorphisms. We will refer to this category as $\SympI$.

In both quantum mechanics and classical mechanics, a pair of physical systems can be treated as a single composite system. The categorical consequence of this fact is that both $\HilbI$ and $\SympI$ are {\em symmetric monoidal} categories (see \cite[Section~5]{coecke}). The monoid structures of $\HilbI$ and $\SympI$ are summarized in Table~\ref{monoid}.
\begin{table}
\begin{tabular}{>{\em}l|ll}
Theory & Quantum mechanics & Classical mechanics \\
Category & $\HilbI$ & $\SympI$ \\
Objects & Complex Hilbert spaces & Symplectic manifolds \\
Arrows & Unitary maps & Symplectomorphisms \\
Monoid operation & Tensor product & Product of manifolds \\
Unit object & The complex numbers & The manifold with one point
\end{tabular}
\caption{The monoid structures of $\HilbI$ and $\SympI$, summarized.}\label{monoid}
\end{table}

The concept of ``the state of a system'' plays an important role in both quantum mechanics and classical mechanics. In order to talk about the state of a system in purely category-theoretic terms, we must enlarge the categories $\HilbI$ and $\SympI$ by adding more arrows to each one. We will begin with some definitions.
\begin{defn}\label{lin.iso}
Let $H$ and $K$ be complex Hilbert spaces. A {\em linear isometry} from $H$ to $K$ is a linear map $U \maps H \to K$ with the property that
\[ \ip{U(\psi)}{U(\tilde{\psi})} = \ip{\psi}{\tilde{\psi}} \]
for all $\psi, \tilde{\psi} \in H$.
\end{defn}
\begin{defn}\label{symp.map}
Let $M$ and $N$ be symplectic manifolds, with symplectic forms $\omega$ and $\sigma$ respectively. A {\em symplectic map} from $M$ to $N$ is a smooth map $\phi \maps M \to N$ with the property that
\[ \sigma_{\phi(x)}(d\phi_x(v), d\phi_x(\tilde{v})) = \omega_x(v, \tilde{v}) \]
for all $x \in M$ and all $v, \tilde{v} \in T_xM$.
\end{defn}
For the remainder of this section, our category-theoretic setting for quantum mechanics will be the category whose objects and arrows are complex Hilbert spaces and linear isometries, which we will refer to as $\HilbE$. Our setting for classical mechanics will be the category whose objects and arrows are symplectic manifolds and symplectic maps, which we will refer to as $\SympE$. For the sake of familiarity, all the definitions and propositions up to this point have been phrased in terms of unitary maps and symplectomorphisms, as if we were working in $\HilbI$ and $\SympI$. The propositions remain true, however, if we work in $\HilbE$ and $\SympE$ instead, reading ``unitary map'' as ``linear isometry'' and ``symplectomorphism'' as ``symplectic map.''

In both $\HilbE$ and $\SympE$, there is a natural one-to-one correspondence between states of a system and arrows from the unit object to the object representing the system. In $\HilbE$, a state is represented by a unit vector $\psi$ in a complex Hilbert space $H$, which corresponds to the unique linear isometry from $\C$ to $H$ that sends $1$ to $\psi$. In $\SympE$, a state is represented by a point $x$ in a symplectic manifold $M$, which corresponds to the unique symplectic map from the manifold with one point to $M$ that sends the one point to $x$.

We are now ready to give a ``traditional'' definition of cloning that makes sense in any symmetric monoidal category.
\begin{defn}[Traditional definition of cloning---categorical]\label{ta.def}
Let $\mathcal{C}$ be a symmetric monoidal category with monoid operation $\circledast$ and unit object $I$. A {\em cloning process} for an object $A$ of $\mathcal{C}$ consists of an arrow $\beta \maps I \to A$ and an arrow $c \maps A \circledast A \to A \circledast A$ such that the following diagram commutes for all $\psi \maps I \to A$:
\[ \xymatrix @!R=0.8in @!C=0.4in {
A \circledast A \ar[rr]^c & & A \circledast A \\
& I \circledast I \ar[ul]^{\psi \circledast \beta} \ar[ur]_{\psi \circledast \psi}
} \]
\end{defn}
When $\mathcal{C}$ is $\HilbE$, Definition~\ref{ta.def} reduces to Definition~\ref{tq.def} with the words ``unitary map'' replaced by ``linear isometry.'' When $\mathcal{C}$ is $\SympE$, Definition~\ref{ta.def} reduces to Definition~\ref{tc.def} with the word ``symplectomorphism'' replaced by ``symplectic map.''

We can also give an ``expanded'' definition of cloning that makes sense in any symmetric monoidal category.
\begin{defn}[Expanded definition of cloning---categorical]\label{ea.def}
Let $\mathcal{C}$ be a symmetric monoidal category with monoid operation $\circledast$ and unit object $I$. A {\em cloning process} for an object $A$ of $\mathcal{C}$ consists of an arrow $\beta \maps I \to A$, an object $B$ of $\mathcal{C}$, an arrow $\rho \maps I \to B$, and an arrow $c \maps A \circledast A \circledast B \to A \circledast A \circledast B$ such that the following diagram commutes for all $\psi \maps I \to A$, where $f$ is some function from $\Hom(I, A)$ to $Hom(I, B)$:
\[ \xymatrix @!R=0.8in @!C=0.4in {
A \circledast A \circledast B \ar[rr]^c & & A \circledast A \circledast B \\
& I \circledast I \circledast I \ar[ul]^{\psi \circledast \beta \circledast \rho} \ar[ur]_{\psi \circledast \psi \circledast f(\psi)}
} \]
\end{defn}
When $\mathcal{C}$ is $\HilbE$, Definition~\ref{ea.def} reduces to Definition~\ref{eq.def} with the words ``unitary map'' replaced by ``linear isometry.'' When $\mathcal{C}$ is $\SympE$, Definition~\ref{ea.def} reduces to Definition~\ref{ec.def} with the word ``symplectomorphism'' replaced by ``symplectic map.'' Notice, also, that Definition~\ref{ea.def} reduces to Definition~\ref{ta.def} when $B$ is the unit object. Consequently, any proposition about Definition~\ref{ea.def} cloning processes applies to Definition~\ref{ta.def} cloning processes as well.
\subsection{Relationship with uniform cloning}
The categorical definitions of cloning given above differ in several ways from the {\em uniform cloning} concept used in Abramsky's categorical no-cloning theorem (\cite[Section~4.1]{abramsky}). One difference is that the definition of uniform cloning makes no reference to arrows from the unit object $I$, which play the role of system states in Definitions \ref{ta.def} and \ref{ea.def}. Another difference is that uniform cloning demands the existence of cloning processes for all of the objects in a category, while Definitions \ref{ta.def} and \ref{ea.def} can be applied to individual objects. These differences make it hard to find a connection between uniform cloning and the definitions of cloning used in this paper.
\section{Directions for further research}
The results presented in this paper are rather basic, but the following questions suggest that further study of classical cloning may yield deeper results, of interest to mathematicians and physicists alike.
\subsection{Is there a cloning process for every symplectic manifold?}
Example~\ref{general} in Section~\ref{cloning-examples} gives an example of a cloning process for every symplectic vector space. Reading this, Bruno Le Floch asked whether there is a cloning process for every symplectic manifold. If not, what are the obstructions to the existence of a cloning process?
\subsection{What are the limitations on cloning in Poisson mechanics?}
Most of the classical systems and dynamics of interest to physicists can be described in the category of symplectic manifolds and symplectic maps. It is often useful, however, to study classical mechanics in the more general category of Poisson manifolds and Poisson maps. One difference between symplectic manifolds and Poisson manifolds is that a symplectic structure is described by a covariant tensor, while a Poisson structure is described by a contravariant tensor. Because of this minor-looking but important distinction, our proof of Proposition~\ref{size} (which says that a cloning machine must be at least as complicated as the object to be cloned) cannot be straightforwardly generalized from symplectic mechanics to Poisson mechanics. Is there a result in Poisson mechanics analogous to Proposition~\ref{size}?
\subsection{What can be said about classical universal constructors?}
The concept of a cloning process is closely related to von Neumann's concept of a universal constructor. A referee pointed out a recent paper \cite{constructor} that formalizes the notion of a universal constructor in quantum mechanics, and shows that a quantum universal constructor would be heavily restricted in its capabilities. What can be said about the existence and limitations of universal constructors in classical mechanics?
\section*{Acknowledgments}
Many thanks to John Baez. Without his advice and encouragement, this paper would not have been written.
\bibliographystyle{plain}
\bibliography{cloning-paper}
\end{document}